\documentclass[11pt]{article}

\usepackage{enumitem}
\usepackage{epsfig}
\usepackage[mathscr]{euscript}
\usepackage{csquotes}

\usepackage[linesnumbered, ruled,noend,noline]{algorithm2e}

\usepackage{float}
\usepackage{graphicx}
\usepackage{ stmaryrd }
\usepackage{multirow}
\usepackage{gensymb}
\usepackage{array}
\usepackage{listings}
\usepackage{amsthm}

\renewcommand\thesection       {\arabic{section}}

\usepackage{amssymb}
\usepackage{graphicx,amsmath,calc}

\newtheorem{theorem}{Theorem}[section]

\newtheorem{corollary}[theorem]{Corollary}

\numberwithin{equation}{section}

\title{Information Accessibility Limits in Structured NP Search}
\author{Jing-Yuan Wei\thanks{Zhejiang Yi-Neng Grid-Storage Energy Co. Ltd. 
5277 Ouyang Rd, Haijing District, Wenzhou, Zhejiang, China. Email: weijingyuan@gmail.com}
}
\date{May 2026} 

\begin{document}
\maketitle

\begin{abstract}
We study the problem of locating violating principal minors in matrix
families lying near the boundary of P-matrices. Rather than viewing
this search problem purely through computational complexity, we
analyze it from an information-accessibility perspective.

We show that, despite strong underlying algebraic structure, the
location of a violating subset may remain difficult to infer through
local queries. In the sparse-violation regime, local observations
typically provide only weak eliminative power, and polynomially many
queries accumulate only vanishing mutual information about the hidden
witness under the induced oracle model.

Using mutual information and Fano's inequality, we characterize the
resulting limitation on information acquisition. The analysis
highlights a conceptual distinction between structure and
accessibility: a problem may possess rich underlying structure while
the information required to identify a hidden witness remains weakly
inferable from observable responses.
\end{abstract}


\section{Introduction}

The study of P-matrices - matrices whose principal minors are all
positive - plays a central role in optimization, complementarity
problems, and numerical analysis. A fundamental computational question
is to determine whether a given matrix is a P-matrix.
Coxson~\cite{Coxson1994} showed that this decision problem is
co-NP-complete, suggesting that the associated search problem of
locating a violating principal minor may also be computationally
difficult.

In this work, we study this search problem from an
information-accessibility perspective. While the decision problem asks
whether a violation exists, the search problem requires identifying a
violating subset among an exponentially large family of principal
submatrices. Our goal is not to establish unconditional hardness
results, but rather to investigate how representation and interaction
govern access to hidden witness information.

Our central observation is that structure alone does not necessarily
yield informative local observations. Even when a matrix exhibits
strong algebraic regularity, the information required to isolate a
violating principal minor may remain distributed across the matrix
structure and only weakly inferable through local queries. Under such
interaction constraints, search may become difficult because the rate
of useful information acquisition remains low despite the presence of
substantial structure.

To investigate this phenomenon, we introduce a structured family of
rank-one perturbations of P-matrices:
\begin{equation}
A(u,v)=M+uv^{\top},
\label{eq:As}
\end{equation}
where \(M\in\mathbb{R}^{n\times n}\) is a P-matrix and
\(u,v\in\mathbb{R}^{n}\). In an appropriate parameter regime, all
principal minors remain positive, while small perturbations can
introduce sparse violations without destroying the underlying
structure.

We analyze the problem under restricted interaction models in which
algorithms access the matrix only through local queries. This
abstraction does not artificially remove structure from the problem,
but instead isolates how the chosen interaction model governs access
to hidden witness information. Under a suitable distribution over hidden violating subsets,
polynomially many queries accumulate only limited mutual information
about the location of the hidden violation. Consequently,
even highly structured instances may require many interactions before
sufficient information about the violating subset becomes inferable.

More broadly, our analysis suggests an information-accessibility
framework for search problems, in which complexity depends not only on
the existence of structure, but also on how effectively that
structure can be exploited through interaction. From this viewpoint,
computational difficulty may arise when the information required to
specify a solution remains only weakly inferable from observable
responses.

Although the present work is motivated more broadly by combinatorial
search problems, we focus here on a structured matrix family in order
to isolate the role of information accessibility. To this end, we
adopt a uniform prior over the hidden witness, motivated by the
absence of strong local inferability among principal minors in general
matrix families (see \ref{app:conditional-minors}). This provides a
clean baseline in which information about the hidden witness must be
acquired through interaction.

This perspective aligns with recent information-theoretic approaches
to search, such as the psocid framework~\cite{Wei2026}, which studies
search limitations under restricted access in structureless settings.
The present work extends this viewpoint to structured problems,
suggesting that strong underlying structure may still produce only
weakly informative local observations under certain interaction
models.

Section~\ref{sc:preliminary} introduces the preliminaries.
Section~\ref{sc:lowerbound} presents the information-theoretic
analysis. Section~\ref{sc:accessibility} interprets the results
through statistical indistinguishability.
Section~\ref{sc:concluding} concludes the paper.

\section{Preliminaries}
\label{sc:preliminary}

Let \(M \in \mathbb{R}^{n \times n}\). For any subset
\(\alpha \subseteq [n]\), let \(M_\alpha\) denote the principal
submatrix indexed by \(\alpha\). A matrix \(M\) is called a P-matrix if
\[
\det(M_\alpha) > 0
\quad
\text{for all }
\alpha \subseteq [n],
\ \alpha \neq \emptyset.
\]

We consider the following search problem.

\paragraph{Violation Search Problem.}
For a P-matrix \(M\), let \(A(u,v)\) be defined
as in~\eqref{eq:As}, where \(u,v \in \mathbb{R}^{n}\).
The problem is to find a subset
\(\alpha \subseteq [n]\) such that
\[
\det(A(u,v)_\alpha)\le 0,
\]
if such a subset exists.

This is the search counterpart of the co-NP-complete problem of
recognizing P-matrices~\cite{Coxson1994}.

We are particularly interested in regimes in which violations of the
P-matrix property are sparse relative to the total number of principal
minors. Let \(A=A(u,v)\) for simplicity and 
\begin{equation} \label{eq:vset}
\mathcal{V}(A)
:=
\{
\alpha \subseteq [n]
:
\det(A_\alpha)\le 0
\}
\end{equation}
denote the set of violating principal minors.

We consider instances satisfying
\[
|\mathcal{V}(A)| \ll 2^n,
\]
so that only a small fraction of principal minors violate the
P-matrix condition.

A particularly informative regime is the
\emph{single-violation regime}, in which
\[
|\mathcal{V}(A)| = 1.
\]
In this case, there exists a unique subset \(w^\star\) such that
\[
\det(A_{w^\star}) \le 0,
\]
while all other principal minors remain positive.

Under this regime, the search problem reduces to locating a single
hidden subset \(w^\star\) among \(2^n-1\) candidates. Although the
matrix retains strong global algebraic structure, the violating witness
remains sparsely embedded within an exponentially large family of
principal minors.

\paragraph{Remark.}
\ref{app:matrixofsix} provides an explicit \(6\times6\)
example satisfying
\(
|\mathcal{V}(A)| = 1,
\)
while \ref{app:construction} gives a construction achieving
the single-violation regime for general P-matrices \(M\).

From the information-accessibility perspective, this regime is
particularly interesting because the global P-matrix structure remains
largely intact while the violation is weakly locally inferable. As a
result, local observations may reveal only limited information about
the location of the hidden violating subset.

\section{Information Accessibility and Search Complexity}
\label{sc:lowerbound}

We formalize our analysis through an information-theoretic perspective
on search problems, which we refer to as the
\emph{Information-Accessibility Framework}.

\paragraph{Information-Accessibility Framework.}
A search problem is specified by a hidden variable \(W\) representing
the solution, together with an interaction model that determines how an
algorithm acquires information about \(W\). The behavior of the search
is governed by two quantities:
\begin{itemize}
\item the \emph{information requirement}, measured by the entropy
\(H(W)\);
\item the \emph{information acquisition rate}, determined by the
information revealed through the interaction model.
\end{itemize}
The central question is whether the accessible information suffices to
identify \(W\).

\subsection{Instantiation: P-Matrix Violation Search}

We instantiate this framework in the setting of P-matrix violation
search. For notational convenience, we write \(N\) in place of \(n\)
in the remainder of the paper.

Let \(A=A(u,v)\) be drawn from the perturbation model defined
in~\eqref{eq:As}, and let \(\mathcal{V}(A)\) denote the set of
violating subsets defined in~\eqref{eq:vset}.
In the single-violation regime,
\(
|\mathcal{V}(A)|=1,
\)
and we denote the unique violating subset by
\(W\subseteq [N]\).

Assume that the hidden violating subset
\(W\in\mathcal{V}(A)\) is uniformly distributed over the admissible
nonempty subsets of \([N]\). The rationale for this non-informative
prior is discussed in \ref{app:conditional-minors}, which analyzes
weak conditional inferability between principal minors in general
matrix families. Hence,
\[
H(W)=\log_2(2^N-1)=\Theta(N),
\]
where entropy is measured in bits.

Although \(A\) is algebraically structured, the condition defining
\(\mathcal{V}(A)\) depends on coordinated global interactions among its
entries, and the location of \(W\) is not directly revealed by local
properties of \(A\).

\paragraph{Induced oracle model.}
We consider an induced oracle model in which the algorithm interacts
with the instance through a sequence of local queries. This interaction
may be viewed as a communication channel through which information
about \(W\) is gradually revealed.

To formalize information accessibility, the algorithm issues queries of
the form: for a subset \(\alpha\subseteq[N]\), observe
\[
Y:=
[\det(A_\alpha)\le 0]
\in\{0,1\}.
\]

The algorithm may adaptively select queries based on past
observations. Formally, at step \(t\), the query \(\alpha_t\) is a
function of the transcript
\(
T_{t-1}=(Y_1,\ldots,Y_{t-1}),
\)
and the response is
\[
Y_t=
[\det(A_{\alpha_t})\le 0].
\]

In the single-violation regime, this oracle satisfies
\[
Y=1
\quad\Longleftrightarrow\quad
\alpha=W,
\]
so the interaction reduces to equality testing.

We emphasize that the algorithm may have direct access to the entries
of \(A\) and may compute principal minors explicitly. The induced
oracle model abstracts this process by recording only the resulting
binary responses associated with queried principal minors, reflecting
their sequential evaluation one by one.

Let \(w^\star\) denote the realization of \(W\). We allow
\(p(N)\) parallel queries per round, where \(p(N)\) is polynomially
bounded, and let \(T\) denote the random round of the first successful
query. In the single-violation regime under the uniform prior, the
determinant-sign query behaves as an equality probe for \(W\). Hence,
by the psocid framework~\cite{Wei2026}, the expected first-hit time
satisfies
\[
\mathbb{E}[T]
=
\Omega\!\left(\frac{2^N}{p(N)}\right).
\]

\paragraph{Remark (Uniform prior and Yao's principle).}
The uniform prior over \(W\) represents maximal uncertainty about the
location of the violating subset and yields
\[
H(W)=\Theta(N).
\]
Within the induced interaction model, Yao's minimax
principle~\cite{Yao1977} provides a standard connection between
distributional analysis and the behavior of randomized algorithms
under the same access assumptions. In particular, hard distributions
may be used to study limitations on information acquisition for the
worst-case behavior of randomized algorithms under the induced
interaction model.

\subsection{Information-Acquisition Limitation}

The following theorem, adapted from the psocid
framework~\cite{Wei2026}, characterizes the limitation on the rate at
which information about \(W\) can be acquired.

\begin{theorem}[Information-acquisition limitation]
\label{thm:capacity}
Under the induced oracle model and a uniform prior over the nonempty
subsets \(W\subseteq[N]\), polynomially many queries accumulate only
vanishing mutual information about the hidden violating subset.

More precisely, for any possibly adaptive randomized algorithm making
at most polynomially many queries, let \(\mathcal{F}\) denote the
resulting random transcript. Then the mutual information between the
hidden violating subset \(W\) and the transcript \(\mathcal{F}\)
satisfies
\[
I(W;\mathcal{F})=o(1).
\]
\end{theorem}

\begin{proof}[Proof sketch]
Consider any adaptive sequence of queries. We allow
\(p(N)\) parallel queries per round, where \(p(N)\) is polynomially
bounded. Let
\(
Y_t=(Y_{t,1},\dots,Y_{t,m})
\)
denote the vector of responses in round \(t\), where
\(m\le p(N)\).

All query outcomes are flattened into a single sequence of
\(q=\mathrm{poly}(N)\) scalar responses:
\begin{equation}
\label{eq:fq}
\mathcal{F}_q=(y_1,\dots,y_q),
\end{equation}
where the variables \(y_k\) enumerate the coordinates
\(Y_{t,j}\).

In the single-violation regime, each query \(\alpha\) returns
\[
Y=[\alpha=W].
\]

Under the uniform prior over the nonempty subsets of \([N]\), and
conditioned on any transcript of \(k-1\) failed queries, the posterior
distribution of \(W\) remains uniform over the remaining candidates.
Hence, for the \(k\)-th query,
\[
\Pr(y_k=1\mid y_{<k})
\le
\frac{1}{2^N-1-(k-1)}.
\]

Since \(k\le q=\mathrm{poly}(N)\), this probability remains
exponentially small. Thus each observation is Bernoulli with parameter
\begin{equation} \label{eq:pk}
p_k
\le
\frac{1}{2^N-O(q)},
\end{equation}
and its binary entropy satisfies
\[
H(y_k\mid y_{<k})
=
-p_k\log_2 p_k
-(1-p_k)\log_2(1-p_k)
=
O\!\left(
\frac{N}{2^N}
\right).
\]

By the chain rule,
\[
I(W;\mathcal{F}_q)
=
\sum_{k=1}^{q}
I(W;y_k\mid y_{<k})
\le
\sum_{k=1}^{q}
H(y_k\mid y_{<k})
=
O\!\left(
\frac{qN}{2^N}
\right)
=
o(1).
\]

Thus polynomially many queries accumulate only vanishing information
about \(W\) within the induced oracle model.
\end{proof}

\begin{corollary}[Limited recovery under polynomial interaction]
By Fano's inequality, any algorithm making polynomially many queries
within the induced oracle model does not accumulate sufficient mutual
information to identify \(W\) with constant success probability.
\end{corollary}

\paragraph{Interpretation.}
Theorem~\ref{thm:capacity} suggests that, within the induced oracle
model, the primary limitation arises from restricted information
acquisition rather than computational effort itself. Even under
unbounded computation, polynomially many interactions reveal only
vanishing information about the hidden witness. Thus, strong
algebraic structure does not necessarily translate into informative
local observations under restricted interaction models.

\paragraph{Intrinsic information limitation.}
One might attribute the difficulty of detecting a negative principal
minor in the near-boundary regime to numerical precision. However, our
analysis assumes exact arithmetic and abstracts away finite-precision
effects.

The rank-one construction clarifies that the difficulty lies not in
creating a violation - a small perturbation suffices - but in acquiring the information needed to locate such
a subset.

The information-accessibility barrier is therefore intrinsic: even under
idealized computation, recovery is limited by the information revealed
through queries rather than by computational power or numerical
precision.

\section{Statistical Indistinguishability and Information Accessibility}
\label{sc:accessibility}

The information-accessibility barrier can be interpreted through the
statistical behavior of the response transcript.

We work in the single-violation regime, where each instance has a
unique violating subset \(W\). For any possibly adaptive randomized
algorithm making
\(
q=\mathrm{poly}(N)
\)
queries, let
\[
\mathcal{F}_q := (y_1,\dots,y_q)
\]
denote the resulting response transcript defined in~\eqref{eq:fq}.

Under the uniform prior over the nonempty subsets of \([N]\), by
\eqref{eq:pk} and the union bound,
\[
\Pr(\exists k \le q : y_k = 1)
\le
\sum_{k=1}^q p_k
\le
\frac{q}{2^N-O(q)}
=
o(1).
\]
Hence, with probability \(1-o(1)\),
\[
\mathcal{F}_q=(0,0,\dots,0).
\]

Consequently, for any two distinct violating subsets
\(w\neq w'\), the corresponding transcript distributions
\[
P(\mathcal{F}_q\mid W=w)
\qquad\text{and}\qquad
P(\mathcal{F}_q\mid W=w')
\]
both assign probability \(1-o(1)\) to the same transcript
\((0,\dots,0)\). Thus, the transcript distributions induced by
different hidden witnesses become asymptotically indistinguishable
under polynomially many interactions.

Equivalently, the transcript carries only vanishing information about
the identity of the hidden violating subset:
\[
I(W;\mathcal{F}_q)=o(1).
\]

\medskip
\noindent\textbf{Interpretation.}
Under the induced interaction model, different candidate witnesses
produce almost identical observable behavior except on exponentially
rare events. Consequently, the observable signal distinguishing
different hidden violating subsets is extremely weak.

In sparse near-boundary regimes, small perturbations may switch the
matrix between feasibility and violation while affecting observable
responses only slightly. As a result, different candidate witnesses
may induce highly similar transcript distributions under restricted
interaction.

This suggests that sparse near-boundary regimes may naturally exhibit
weak inferability of hidden witness information from observable
responses.

\medskip
\noindent\textbf{Discussion.}
The present analysis applies specifically to the induced oracle model
studied in this work. Nevertheless, the sparse near-boundary regime
considered here may reflect a broader phenomenon in structured search
problems: the existence of a witness may be detectable while the
information required to identify a specific witness remains only
weakly inferable from observable responses.

In such settings, individual observations may provide only weak
eliminative or directional information, leading to statistical
indistinguishability among candidate witnesses even in the presence
of substantial underlying structure.

From this viewpoint, sparse near-boundary regimes may provide a useful
conceptual setting for studying how weak inferability and restricted
interaction jointly limit information acquisition in structured search
problems.

\section{Concluding Remarks}
\label{sc:concluding}

We study P-matrix violation search from an
information-accessibility perspective. The analysis suggests that even
highly structured matrix families may exhibit weak inferability of
hidden witnesses under restricted interaction models. In the
sparse-violation regime, local observations provide only limited
eliminative power, so identifying the violating subset requires
aggregating information across a large family of candidates.

More broadly, the present work highlights a conceptual distinction
between structure and accessibility. A problem may possess rich
internal structure while the information required to identify a hidden
witness remains difficult to extract through the available
interaction model. From this viewpoint, computational difficulty may
arise not from the absence of structure itself, but from limitations
on information acquisition through interaction.

A key assumption of the present analysis is the use of a simplified
interaction model and a non-informative prior, motivated by the weak
conditional inferability observed among principal minors in general
matrix families. In other NP search problems, however, such as 3-SAT,
clause overlap and adaptive inference introduce additional
dependencies that may substantially improve local inferability.
Algorithms such as PPSZ~\cite{PPSZ1998} demonstrate that structural
features can be exploited to improve search performance through local
propagation and inference.

This local inferability advantage may involve a complementary
tradeoff. Representational-expansion effects studied
in~\cite{Wei2026-1} suggest that improvements in local inferability
may arise at the cost of enlarging the representation through
auxiliary variables and consistency structures.

The present work suggests that information accessibility
may provide a useful complementary viewpoint for structured search
problems. The analysis indicates that substantial underlying structure
does not necessarily imply strong local inferability of hidden witness
information under restricted interaction models. Understanding how
interaction constraints shape information acquisition and witness
recoverability may therefore provide a useful direction for future
research.

\appendix
\renewcommand{\thesection}{Appendix~\Alph{section}}

\section{Conditional Accessibility Between Principal Minors}
\label{app:conditional-minors}

We briefly explain why the uniform prior used for the hidden violating
subset is a natural non-informative model in the absence of additional
structural assumptions.

The key point is that a uniform prior over the location of the hidden
violating subset does not require probabilistic independence among all
principal minors. Rather, it reflects the absence of strong local
inferability: observations of one principal minor generally do not
provide sufficient information to strongly privilege or eliminate
other candidate subsets.

Consider two principal minors whose index sets overlap. After a
suitable permutation of rows and columns, a principal submatrix can be
written in block form as
\[
A=
\begin{pmatrix}
C & D\\
E & F
\end{pmatrix},
\]
where \(F\) corresponds to the common part shared with another minor.
When \(F\) is nonsingular, the Schur complement formula gives
\[
\det(A)
=
\det(F)\,
\det(C-DF^{-1}E).
\]

Thus, even if the sign of \(\det(F)\) is known, the sign of
\(\det(A)\) additionally depends on the Schur-complement factor
\[
\det(C-DF^{-1}E).
\]
This factor involves interactions between the overlapping and
non-overlapping coordinates and is not determined by the sign of
\(\det(F)\) alone. Consequently, knowing that one principal minor is
positive does not, in general, determine whether another overlapping
principal minor is positive or nonpositive, let alone the signs of
more distant minors. Thus, in general matrix families, the sign of one queried principal
minor typically does not provide sufficient information to eliminate
other unqueried minors without further computation.

Importantly, the uniform prior concerns the location of the hidden
violating subset rather than probabilistic independence among all
principal minors. Even when the principal minors themselves are
strongly correlated, weak local inferability means that no candidate
subset is strongly privileged before interaction.

An analogy may help clarify this distinction. Suppose there are
\(
K=2^N-1
\)
iron balls of identical size but different weights, and exactly one
ball weighs less than one gram. The ball weights may be highly
correlated and therefore need not be mutually independent.
Nevertheless, before any measurement, no ball is distinguished as
more likely than another to be the uniquely light ball. Weighing one
ball reveals whether that particular ball is light, but generally does
not determine the weights of the remaining balls. Thus, after
observing that one tested ball is not the special one, the remaining
candidates retain essentially symmetric status.

The single-violation regime considered in the present work is similar.
There are
\[
K=2^N-1
\]
candidate principal minors, and exactly one violating subset
\(W\). Observing that one queried principal minor is positive
typically removes only that queried candidate itself from
consideration, while the remaining candidates continue to possess
essentially symmetric status in the absence of additional structural
information.

From this viewpoint, the hidden violating subset retains essentially
uniform status among the remaining candidate principal minors, while
local sign observations provide only weak eliminative power. This
motivates the use of a uniform prior in the main text for general
matrix families in which principal minors exhibit weak eliminative
inferability.

The situation may differ for special matrix classes, such as
\(M\)-matrices, \(Z\)-matrices, totally positive matrices, or other
highly structured families, where algebraic constraints may induce
stronger dependencies and more informative eliminative relations among
principal minors. Thus, the uniform prior adopted in the present work
should be viewed as a baseline model for general matrix regimes,
rather than as a universal statement applying to all matrix classes.

\section{A \(6\times6\) Illustrative Example}
\label{app:matrixofsix}

In this appendix, we present a concrete example illustrating the
perturbation model introduced in the main text. The example exhibits a
\emph{single-violation regime}, in which a structured matrix admits
exactly one non-positive principal minor under a rank-one perturbation.

\paragraph{Base P-matrix.}
Consider the matrix
\[
M=
\begin{pmatrix}
2.3 & -4.8 & -1.9 & 0 & 0 & 1.3 \\
2.6 & 3.4 & 3.8 & 0 & 0 & 2.6 \\
-2.6 & 4.8 & 6.7 & -1.1 & 1.3 & -1.3 \\
1.3 & 2.4 & -3.8 & 5.4 & 0 & -1.3 \\
-2.6 & 4.8 & 0 & -1.1 & 4.9 & -1.3 \\
0 & -4.8 & -1.9 & 0 & -2.6 & 7.5
\end{pmatrix}.
\]

Numerical verification shows that all \(63\) nonempty principal minors
of \(M\) are positive. In particular, the smallest principal minor is
approximately \(0.272\), attained at the subset \(\{2,3,4\}\).
Thus, \(M\) is a P-matrix.

\paragraph{Perturbation.}
Let
\[
u=(0.25,\;0,\;0,\;0,\;2,\;0)^{\top},
\qquad
v=(-1.740695,\;0,\;0,\;0,\;-1.740695,\;0)^{\top}.
\]

\paragraph{Violation.}
For this choice of \(u\) and \(v\), numerical evaluation shows that
exactly one principal minor of \(A(u,v)\), defined
in~\eqref{eq:As}, is non-positive. Specifically, the principal
submatrix indexed by
\[
w^\star=\{1,5\}
\]
satisfies
\[
\det\bigl(A(u,v)_{w^\star}\bigr)\approx -0.001.
\]

All other nonempty principal minors of \(A(u,v)\) remain positive.
For comparison,
\[
\det\bigl(M_{w^\star}\bigr)=11.27,
\]
which is substantially larger than the minimal principal minor
\(0.272\) of the original P-matrix \(M\).

\begin{table}[h]
\centering
\footnotesize
\caption{\footnotesize Principal minors of neighboring subsets of the violating subset
\(\{1,5\}\).}
\vspace{0.15cm}
\label{tab:neighboring-minors}
\begin{tabular}{c c r @{\hspace{0.6cm}} c c r @{\hspace{0.6cm}} c c r}
\hline
No. & Subset & Minor ~~\vline&
No. & Subset & Minor ~~\vline&
No. & Subset & Minor \\
\hline
1  & \(\{1\}\)       & 1.865  ~~\vline &
6  & \(\{1,6\}\)     & 13.986 ~~\vline &
11 & \(\{5,6\}\)     & 7.260   \\

2  & \(\{1,2\}\)     & 18.820  ~~\vline&
7  & \(\{5\}\)       & 1.419  ~~\vline &
12 & \(\{1,2,5\}\)   & 12.270  \\

3  & \(\{1,3\}\)     & 7.554   ~~\vline&
8  & \(\{2,5\}\)     & 4.823   ~~\vline&
13 & \(\{1,3,5\}\)   & 8.006   \\

4  & \(\{1,4\}\)     & 10.070 ~~\vline &
9  & \(\{3,5\}\)     & 9.505   ~~\vline&
14 & \(\{1,4,5\}\)   & 0.617   \\

{\bf 5}  & \({\bf \{1,5\}}\)     & {\bf -0.001}  ~~\vline&
10 & \(\{4,5\}\)     & 7.660   ~~\vline&
15 & \(\{1,5,6\}\)   & 14.244  \\
\hline
\end{tabular}
\end{table}

\begin{figure}[H]
\centering
\includegraphics[scale=1.0]{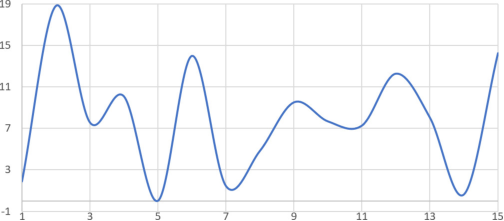}
\vspace{-0.3cm}
\caption{\label{fig:01} {\footnotesize Principal minors of neighboring subsets of the violating subset
\(\{1,5\}\) indexed according to Table~1.}}
\end{figure}

Table~\ref{tab:neighboring-minors} and Figure~\ref{fig:01} summarize
the principal minors associated with subsets neighboring the violating
subset \(\{1,5\}\). Although \(\{1,5\}\) is the unique violating
subset, most nearby subsets remain positive and several have
substantially larger determinant values. In particular, no clear
monotone trend toward the violation is visible among neighboring
subsets.

\paragraph{Remark.}
This example illustrates that a highly structured matrix, obtained
through a smooth rank-one perturbation of a P-matrix, can exhibit a
sparse and weakly expressed violation of the P-matrix property. Even
among neighboring subsets, the violating minor is not strongly
distinguished by nearby determinant values. Consequently, local
evaluations of principal minors provide only limited directional
information about the location of the hidden violation.

This behavior contrasts with many constraint-based search problems,
such as 3-SAT, where local constraints often provide informative
signals that guide inference toward a solution. In the present
setting, however, no obvious local pattern among nearby principal
minors clearly identifies the violating subset.

The construction is obtained numerically and is intended only as an
illustrative example demonstrating the existence of sparse-violation
regimes in which the hidden violating subset is not readily accessible
through local observations.

\section{Rank-One Construction of Sparse Violations in P-Matrices}
\label{app:construction}

We describe a procedure for constructing sparse-violation instances,
particularly in regimes where
\(
|\mathcal{V}(A(u,v))|=1.
\)
The key idea is to perturb a P-matrix through a rank-one update so
that one principal minor becomes non-positive while the others remain
positive.

For a P-matrix \(M\in\mathbb{R}^{n\times n}\), define
\[
f_M
:=
\min_{\alpha\subseteq[n],\,\alpha\neq\emptyset}
\det(M_\alpha),
\qquad
\alpha^\star
\in
\arg\min_{\alpha}
\det(M_\alpha).
\]
Thus, \(f_M\) denotes the smallest principal minor of \(M\), and
\(\alpha^\star\) is a corresponding minimizing subset.

\paragraph{Algorithm~1 (Rank-one construction).}
\begin{enumerate}

\item Compute \(f_M\) and a minimizing subset \(\alpha^\star\).
Choose a vector \(u\in\mathbb{R}^n\) satisfying
\[
u_i>0
\quad\text{for } i\in\alpha^\star,
\qquad
u_i=0
\quad\text{otherwise}.
\]

\item Define a vector \(\hat v\in\mathbb{R}^n\) supported on
\(\alpha^\star\), and choose its signs so that
\[
\hat v_{\alpha^\star}^{\top}
M_{\alpha^\star}^{-1}
u_{\alpha^\star}
<0.
\]

\item Consider the one-parameter family
\[
A(\lambda)
:=
M+\lambda u\hat v^{\top},
\qquad
\lambda>0.
\]

By the matrix determinant lemma,
\[
\det(A(\lambda)_{\alpha^\star})
=
\det(M_{\alpha^\star})
\left(
1+
\lambda\,
\hat v_{\alpha^\star}^{\top}
M_{\alpha^\star}^{-1}
u_{\alpha^\star}
\right).
\]

\item Choose \(\lambda>0\) such that
\[
\det(A(\lambda)_{\alpha^\star})
=
-\epsilon,
\]
where \(\epsilon>0\) is small.

\item If necessary, reduce \(\lambda\) slightly to ensure that there
exists a unique subset \(w^\star\) satisfying
\[
\det(A(\lambda)_{w^\star})<0,
\]
while all other principal minors remain positive.
Note that \(w^\star\) need not coincide with \(\alpha^\star\).

\item Set \(v:=\lambda\hat v\) and define
\[
A(u,v):=M+uv^{\top}.
\]

\end{enumerate}

Upon termination, the matrix \(A(u,v)\) satisfies
\[
\det(A(u,v)_{w^\star})<0,
\]
while all other principal minors remain positive, yielding a
single-violation instance.

\paragraph{Remarks.}
\begin{itemize}

\item In favorable regimes, or under sufficiently small perturbations,
the construction may yield
\(
|\mathcal{V}(A(u,v))|=1.
\)

\item Although the construction is initialized using a specific subset
\(\alpha^\star\), identifying the smallest principal minor itself
requires a global search over exponentially many subsets. Thus, the
construction does not assume that the minimizing subset can be located
efficiently from local information.

\item The rank-one perturbation is applied only after identifying a
globally minimal subset. From the information-accessibility
perspective, this highlights an important distinction: once a suitable
subset is known, constructing a sparse violation becomes relatively
straightforward, whereas locating such a subset may itself require
global information.

\item The magnitude of the violating principal minor can be made
arbitrarily small. Consequently, the resulting violation may remain
weakly expressed relative to neighboring positive minors, reducing the
local distinguishability of the hidden violating subset.

\item The example in \ref{app:matrixofsix} employs a different choice
of perturbation vectors \(u\) and \(v\), illustrating that
sparse-violation regimes can arise from a broader class of rank-one
perturbations.

\end{itemize}



\begin{thebibliography}{99}

\bibitem{Coxson1994}
G. E. Coxson,
\emph{The P-matrix problem is co-NP-complete},
Mathematical Programming, 64 (1994), 173--178.

\bibitem{PPSZ1998}
R.~Paturi, P.~Pudl{\'a}k, M.~E.~Saks, and F.~Zane,
\newblock An improved exponential-time algorithm for $k$-SAT,
\newblock in \emph{Proceedings of the 39th Annual IEEE Symposium 
on Foundations of Computer Science (FOCS)}, 1998, pp.~628--637.


\bibitem{Wei2026}
J.-Y.~Wei,
\emph{Intrinsic Information Flow in Structureless NP Search},
arXiv:2603.06315, 2026.

\bibitem{Wei2026-1}
J.-Y.~Wei,
\emph{Information Redistribution Under Reductions in NP Search},
arXiv:2605.20236, 2026.

\bibitem{Yao1977}
A.~C.-C.~Yao,
\emph{Probabilistic computations: Toward a unified measure of complexity},
in \emph{Proceedings of the 18th Annual IEEE Symposium on Foundations of Computer Science (FOCS)},
pp.~222--227, 1977.


\end{thebibliography}
\end{document}